\documentclass[aps,pra,twocolumn]{revtex4-2}
\usepackage{amsmath,amsfonts,amssymb,amsthm}
\usepackage{mathtools}
\usepackage[pdftex]{graphicx}
\usepackage[usenames, dvipsnames, svgnames, table]{xcolor}
\usepackage[colorlinks=true, citecolor=Blue, linkcolor=BrickRed, urlcolor=Brown]{hyperref}
\usepackage{txfonts}
\usepackage{mathrsfs}
\usepackage{bm}
\usepackage{multirow}
\usepackage{braket}
\usepackage{import}
\usepackage{qcircuit}
\usepackage[ruled, noline, noend]{algorithm2e}
\usepackage{array}
\usepackage{comment}
\usepackage{tikz}
\usetikzlibrary{automata,arrows,positioning,calc}
\usepackage{float}

\newtheorem{lemma}{Lemma}

\newtheorem{Note}{Note}

\theoremstyle{definition}
\newtheorem{definition}{Def.}
\newcommand{\be}{\begin{equation}}
\newcommand{\ee}{\end{equation}}
\newcommand{\ben}{\begin{eqnarray}}
\newcommand{\een}{\end{eqnarray}}
\newcommand{\bes}{\begin{subequations}}
\newcommand{\ees}{\end{subequations}}
\newcommand{\bF}{\begin{figure}}
\newcommand{\eF}{\end{figure}}

\DeclareMathAlphabet{\pazocal}{OMS}{zplm}{m}{n}

\newcommand{\orcid}[1]{\href{https://orcid.org/#1}{\includegraphics[height = 2ex]{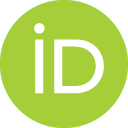}}}

\begin{document}

\title{Where Are All The Tourists From 3025?}

\author{Andrew Jackson \orcid{0000-0002-5981-1604}}
\affiliation{School of Informatics, University of Edinburgh, Edinburgh, EH8 9AB, United Kingdom}
\date{\today}


\begin{abstract}
     This paper examines the distinct lack of clear examples of time-travellers and proposes an explanation for their absence without assuming technical barriers to constructing time machines. Instead, it develops and then analyses a model of the consequences of time-travellers; finding that time travel is self-suppressing.
\end{abstract}

\maketitle
\section{Introduction}
\label{IntroSec}
On 28 June 2009, Stephen Hawking held a party for time-travellers~\cite{Walfisz_2023, Sutter_2023}.
Champagne, balloons, nibbles: this party had all the trappings of an era-defining party, hoping to entice time-travellers into making an appearance. 
However, nobody was invited to the event -- in advance -- and it was only publicly announced afterwards, with exact time and location details, in the hope that an invite would survive the centuries/millennia required to fall into the hands of a party-loving time-traveller.
Despite all the effort put into throwing and publicizing this party, ``no one came"~\cite{Venables2012}. Leaving Hawking to party alone.

Hawking was not the first to try such an experiment~\cite{Block_2005}, which is not too surprising as the experiment is both simple and intuitive: if time travel will \emph{ever} be invented, we should expect our great$^n$-grandkids to come back for a visit and such a visit would prove that time travel is technically and metaphysically possible.
But this has never happened, at least as far as is known~\footnote{If you are visiting from the future and reading this, I hope you had a good laugh at my expense.}. \emph{All} experiments along the same lines as Hawking's have failed to find time-travellers.

These negative results are strange as -- to my knowledge -- there are no explicit bans on time travel in our current understanding of physics~\cite{Blitz_2022, Frankel_2024} and, in fact, the Einstein field equations admit a solution, known as the Gödel metric, with closed-timelike-curves~\cite{RevModPhys.21.447}.
So if physics does not forbid it and, if the human population of the future is as large as seems reasonable to assume, it appears to be a problem that not one of the billions/trillions of future-humans has travelled back to our time and -- intentionally or otherwise -- made themselves known to us.

This leaves us to wonder, can the lack of known tourists from the future be interpreted as proving time travel to be impossible?

In this paper, I argue the contrary. Presenting a model where time travel is theoretically and technologically possible; the construction of time machines is physically, practically, logistically, financially, and politically possible; and nothing prevents those machines from travelling back to our time, but we do not see any travellers from the future.

\section{Preliminary Analysis and Model Setup}
The model I use to analyze time-travel as a phenomenon is derived from a full description of the many possible worlds~\cite{sep-possible-worlds} it is possible to be living in (i.e. the set of worlds that could have been \emph{the} world if things had been different). The worlds are each modeled by a full logical description of them via a -- potentially infinite -- set of logical statements that are true in that world. The worlds are then grouped into sets of worlds based on those that are the same in the regard I care about, the number of time machines that are ever constructed, and this gives rise to the true model I analyse herein. 

\subsection{Ontology of the Model: Timelines, Construction Numbers, and Macrostates}
The foundational notion of the model I now begin building is:\\

``The world is all that is the case."~\cite{AllIsCase}\\

Hence, all possible worlds~\cite{sep-possible-worlds} that I consider will be modeled as complete consistent sets~\cite{2024Complete} of logical statements. 
However, to obtain a viable model of different timelines and different worlds that could have occurred if the past was different, I have to differentiate between necessary and contingent logical statements~\footnote{The difference and relationship between these two kinds of logical statement is deep and beyond the scope of this paper but, regardless of how you choose to define them, the below analysis is the same}.

Necessary logical statements have to be true in every possible world (i.e. in all previously-mentioned complete consistent sets in my model), while the contingent logical statements are free to be true or not in every possible world (provided they are consistent with all necessary logical statements and all contingent logical statements in the same world).  

This distinction between the necessary and contingent truths is captured in the definition of all the possible worlds of the model (Def.~\ref{def:SetofMicros}). But presenting this definition first requires I define the universe of the logical statements that the possible worlds -- in the model used throughout this paper -- consist of, in Def.~\ref{def:Universe}. 
\begin{definition}
    \label{def:Universe}
    The set \underline{$\mathbb{U}$} is the set of all things that may contingently exist i.e. the set of all things that do not necessarily not exist (the set of all things that exist in \emph{some} possible world).
    Logical statements in the sets that constitute my model (i.e. the worlds in $\mathbb{W}$) may only refer to elements of $\mathbb{U}$.
\end{definition}
Then let $\mathbb{N}_{es}$ be the set of all logical statements deemed necessary and let $\mathbb{N}_{es}'$ be its deductive closure, I can then define the possible worlds of the model, in Def.~\ref{def:SetofMicros}.
\begin{definition}
    \label{def:SetofMicros}
    \underline{$\mathbb{W}$} is defined as the set:
    \begin{align}
        \big\{ \mathbb{N}_{es}' \cup \mathbb{C}_{on}  \text{ } \vert \text{ } \mathbb{N}_{es}' \cup \mathbb{C}_{on} \text{ is a complete consistent set over } \mathbb{U} \big\}.
    \end{align}
    \underline{$\mathbb{W}$} is the set of all possible worlds (each described as a set of logical statements) considered in the model.
\end{definition}
\subsection{Grouping Microstates into Macrostates}
\label{sec:Grouping}
In Sec.~\ref{sec:Grouping} I now simplify $\mathbb{W}$ by ignoring the aspects of each world that are irrelevant for my considerations herein. As I only care about the number of time machines constructed in a world. This reduces some worlds to be identical (up to irrelevant and ignored features) and I can effectively consider them as a single world (as defined in Def.~\ref{def:macrostates}).  

As the only feature of worlds I will be concerned with is the number of time machines constructed, \emph{all} worlds can be grouped into sets based on this integer value, which I term the construction number, defined in Def.~\ref{def:contrucNum}. 
\begin{definition}
    \label{def:contrucNum}
    For any timeline, the \underline{construction number of the timeline} is the number of individual time machines constructed in that timeline.    
\end{definition}
The construction number of the timeline enables me to define an equivalence relation on $\mathbb{W}$, with worlds being equivalent if they have the same construction number.
I define the equivalence classes corresponding to this equivalence relation, referred to as macrostates of $\mathbb{W}$, in Def.~\ref{def:macrostates}.
\begin{definition} 
    \label{def:macrostates}
    A \underline{macrostate} of $\mathbb{W}$ is a subset of  $\mathbb{W}$ defined by:
    \begin{align}
         \mathbb{M}_n = \big \{ x \in \mathbb{W} \text{ } \vert \text{ } \mathcal{C}(x) = n \big \},
    \end{align}
    where $n \in \mathbb{N}_0$ and $\mathcal{C}: \mathbb{W} \longrightarrow \mathbb{N}$ maps timelines to their construction number. I.e. a macrostate of $\mathbb{W}$ is a subset of $\mathbb{W}$ containing all timelines with a certain construction number.
\end{definition}
Defining microstates is then much simpler than defining macrostates, in Def.~\ref{def:microstates}.
\begin{definition}
    \label{def:microstates}
    Every timeline in $\mathbb{W}$ is a \underline{microstate} of $\mathbb{W}$. Every microstate is in exactly one macrostate of $\mathbb{W}$.
\end{definition}

\subsection{Dynamics of the Model: Symmetries and State Transitions}
\label{sec:BasConsids}
I now present the rudiments and foundations of how my model functions dynamically.

The first parts have already been presented: the ontology and grouping into macrostates of $\mathbb{W}$. For the rest of this paper and the rest of the construction of this model, I will neglect microstates. Instead, I will only consider macrostates as all the information of interest herein -- that can be extracted from the model -- is contained in knowing what the probabilities of being in each macrostate are.

The key example of this disregard of microstates is how I define transitions between worlds. Notably, I only define transitions between macrostates. No transitions between microstates are ever considered.

Additionally, I assume that, for any fixed macrostate, the only possible transitions are to macrostates with construction numbers that differ from the construction number of the fixed macrostate by one. The transition rates for the possible transitions \emph{out} of the fixed state are all identical. I.e., the transitions probabilities are only non-zero between macrostates where the difference between their construction numbers is exactly one and, in this case, the probability of a transition only depends on the macrostate the transition is \emph{from}\footnote{Meaning if the construction number is increasing or decreasing does not affect the transition rate}.

Furthermore, the transitions from states with greater construction numbers are greater. This reflects how a greater construction number reflects a greater amount of time travel and hence a greater potential for acts that will cause a transition (i.e. changing the time line), e.g. a greater number of butterflies will flap their wings differently and cause different tornadoes. In this metaphor, those tornadoes will turn out to affect the number of butterflies.
The limit of this is when the construction number is zero, when the transition rate out of that macrostate is zero: there are no time machines and hence no way to change the timeline.

A final, informal, detail of my model is that I assume there is a maximum possible construction number a macrostate could possibly have. This limitation follows from the non-zero space and time requirements of constructing any machine, and the presumed finite amount of possible time in a timeline (due to the inevitability of the heat-death of the universe~\footnote{I assume that time-travel does nothing to counter-act the second law of thermodynamics; I cannot know this, but no new invention or technical innovation has threatened the second law since its discovery.})

\subsection{A Formal Statistical Model}
\label{sec:ModelStatement}
To formally define my model, I begin by constructing a Markov chain, where each state of the Markov chain is a macrostate (as defined in Def.~\ref{def:macrostates})~\footnote{I.e. the state space of the Markov chain is the set of all Macrostates} and the transition probabilities are proportional to the transition rates referred to above (scaling by the duration each time step lasts in the Markov chain). The aforementioned time step duration is then allowed to tend to zero, to obtain a continuous time Markov chain and lose the unphysical discretizations.

As in Def.~\ref{def:SetofMicros}, $\mathbb{W}$ is the set of all possible timelines I define $\mathcal{C}: \mathbb{W} \longrightarrow \mathbb{N}_0$ to map timelines to their construction number. Then, I define the states of the Markov chain in my model by: $\forall k \in \mathbb{N}$,
\begin{align}
    \mathbb{M}_k = \big\{ w \in \mathbb{W} \text{ } \vert \text{ }\mathcal{C} (w) = k \in \mathbb{N}_0^{\leq \mathcal{M}_{ax}} \big\},
\end{align}
where $\mathcal{M}_{ax}$ is the maximum possible construction number, referred to before.

Note that these are exactly the macrostates of $\mathbb{W}$.
\newline
As the central mathematical object of my model will begin as a discrete time Markov chain, I pause briefly to define both discrete-time Markov chain (in Def.~\ref{def:Markov}) and continuous-time Markov chains (in Def.~\ref{def:ConMarkov}).
\begin{definition}
    \label{def:Markov}
     A \underline{discrete-time Markov chain} is a sequence of random variables, $X_1, X_2, ..., X_N$, with each $X_j$ taking a value from the state space~\footnote{Which, in the model being constructed herein, is the set of all the macrostates.}.
    The move from one random variable in the sequence to the next is called a step and the probability distribution of $X_j$ depends only on the value of $X_{j-1}$. This property of the sequence is called the Markov property~\footnote{The Markov property can be pithily expressed as: the future is independent of the past, given the present.}.
    A particular discrete time Markov chain is specified by:
    \begin{enumerate}
        \item A state space, the set of possible values a variable can take.
        \item A transition matrix, defining the probabilities of moving from one given state (from the state space) to another in a single step. I.e. the probability distribution of $X_j$ given $X_{j - 1}$.
    \end{enumerate}
\end{definition}
I then treat continuous-time Markov chains informally, for the sake of readability and accessibility, as to treat them formally would raise the statistical complexity of this paper.
\begin{definition}
    \label{def:ConMarkov}
     A \underline{continuous-time Markov chain} can be viewed as a discrete-time Markov chain in the limit where the duration of each step goes to zero. I.e. it now consists of an infinite sequence of random variables, $X_1, X_2, ...$. Importantly, it still obeys the Markov property.
\end{definition}

\begin{figure}[h!]
\begin{center}
	\begin{tikzpicture}[->, >=stealth', auto, semithick, node distance=2.5cm]
	\tikzstyle{every state}=[fill=white,draw=black,thick,text=black,scale=1, minimum size = 3em]
	\node[state]    (A)                     {$\mathbb{M}_{j -1}$};
	\node[state]    (B)[right of=A]   {$\mathbb{M}_{j}$};
	\node[state]    (C)[right of=B]   {$\mathbb{M}_{j+1}$};
	\path
	(B) edge[bend left, below]	node{$\alpha_{j}$}	(A)
     (A) edge[bend left, above]	node{$\alpha_{j-1}$}	(B)
     (C) edge[bend left, below]	node{$\alpha_{j+1}$}	(B)
     (B) edge[bend left, above]	node{$\alpha_{j}$}	(C);
	\end{tikzpicture}
\end{center}
    \caption{A depiction of a subset of nodes in the Markov chain and the transition probabilities (exclusively) between them.\\
    Note that there are, in general, more nodes in this chain that are not shown but do have arrows to and from the nodes $\mathbb{M}_{j+1}$ and/or $\mathbb{M}_{j-1}$.}
    \label{fig:subsetOfMarkovChain}
\end{figure}
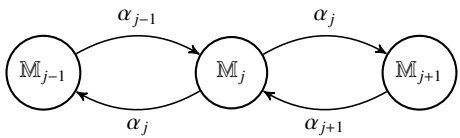

With Markov chains (both discrete and continuous) defined, and the states of the Markov chain in my model defined (as the macrostates of $\mathbb{W}$, defined via the construction number), the only aspect of my model left to specify is the transition probabilities between macrostates (which are defined via transition rates, defined below, to avoid problems when the time step length is allowed to tend to zero). 

All of the conditions required of the dynamics of the model (as described in Sec.~\ref{sec:BasConsids}) are captured by the below Def.~\ref{def:transitionDef}, which defines the transition rates, $\{ \alpha_j \}_{j = 0}^{\infty}$ (most explicitly in Eqn.~\ref{eqn:alphaJ}). For any macrostate with construction number $j \in \mathbb{N}_0$, $\alpha_j$ is the rate of transitions both:
\begin{enumerate}
    \item \emph{from} the macrostate with construction number $j \in \mathbb{N}$ \emph{to} the macrostate with construction number $j + 1 \in \mathbb{N}$
    \item \emph{from} the macrostate with construction number $j \in \mathbb{N}$ \emph{to} the macrostate with construction number $j - 1 \in \mathbb{N}$
\end{enumerate}
\begin{definition}
\label{def:transitionDef}
The transition rate associated with all transitions \emph{from} the macrostate with construction number $j \in \mathbb{N}$ is:
\begin{align}
    \label{eqn:alphaJ}
    \alpha_j
    &=
    \beta \dfrac{j}{2(j + 1)},
\end{align}
where $\beta \in \mathbb{R}^+$ and scales the transition probabilities uniformly across macrostates. Herein, it is left undefined to show it does not materially affect the results.
\end{definition}
I additionally note that the probability of the macrostate with the maximum possible construction number transitioning to a greater macrostate and the probability of the macrostate with transition number zero transitioning to a lesser macrostate is zero, although the other transitions of these macrostates is exactly as they would otherwise be.

\section{Dynamics of Timelines and Stable Timelines}
\label{sec:DynamicsTimelines}
\subsection{A Comparison to Special Relativity}
When looking at the dynamics of time travel and considering motion through time, a sensible starting point is special relativity. A key feature of special relativity is time dilation, where if two people start off together and one goes on a trip at relativistic speeds before reuniting with the other they will each have experienced different amounts of time.

I now ask a similar question. What if two people initially start off together before one goes off travelling in a time machine, while the other stays in place for some time $t \in \mathbb{R}^+$ (according to the stationary person), how long will it have been for the person who went off time-travelling?

The answer is surprisingly un-mathematical: as long as the traveller feels like taking! They could witness as many historical events as they feel like seeing before returning to meet their non-travelling friend.

I try to fit this idea -- conceptually -- into the model of special relativity wherein the traveller and non-traveller experience time passing at different rates due to the effect of their motion. So, the relativistic traveller experiences a factor -- based on their speed -- times as much time as their non-travelling friend (I ignore the twin paradox as general relativity gives the true answer). In this model, the traveller and their friend can be perceived as moving through spacetime at different angles, and the differing durations can be seen as the effect of projecting their path onto the stationary time dimension.

For a time-traveller, the factor that their perception of time is multiplied by (analogously to the above) is, due to them being able to spend as long a time as they want travelling before meeting their friend whenever they like, infinity.

The projection of the time that has passed for them onto their non-travelling friend's time dimension will always be zero and so I term the time that the traveller experiences orthogonal time. No amount of orthogonal time requires more than zero time (as measured by a non-travelling observer).

In fact, regardless of the motion (relativistic or otherwise) of the non-time-traveller, the traveller can always experience any desired about of orthogonal time before returning to meet their non-travelling friend any amount of non-traveller measured time later.

\subsection{Perceptions of the Effects of Time Travel for Non-Travellers}
The discussion of the preceding subsection may initially seem immaterial for our considerations, but when you consider how the actions of time-travellers may change our timeline with us in it this changes rapidly. As they may experience any amount of orthogonal time instantaneously (from our perspective), they may make infinitely many changes to the timeline instantly. If there are many time-traveller travelling about, the combined effect of the entire history (history in orthogonal time, that is) of all time-travel happens instantaneously from our perspective. The combined effects of the timeline being buffeted by changes these travellers intentionally or unintentionally make, changes being made and unmade, again and again, all happen in an instant from our non-travelling perspective. We jump to the final conclusion after an infinite amount of orthogonal time.

So, factoring in the discussion of orthogonal time in Sec.~\ref{sec:DynamicsTimelines}, the only timeline we should ever, as non-travellers, expect to see is the asymptotic limit of orthogonal time.
As the model takes place entirely in orthogonal time, in the upcoming analysis or when I run my simulations the most important feature will be what happens in the limit as orthogonal time tends to infinity. As the only time that means anything when considering the dynamics of timelines (and hence the one I have to use in my model) is orthogonal time. So the limit as the travellers orthogonal time tends to infinity is all non-travellers (i.e. us) will ever see. All the intermediate states of the timeline -- from changes that were undone or themselves changed -- are erased (or more accurately, are prevented from ever having happened) and only the asymptotic limit timeline remains.

\begin{Note}
    I pause to note that the ``time" in the Markov chains in the model I will construct is itself orthogonal time, as just defined.
\end{Note}

\section{Analytic Analysis of the Model}
\subsection{Reformulation as a System of Differential Equations}
I now aim to find the dynamics of the model constructed in the previous section. These dynamics take place as orthogonal time passes.
The most important quantities of the model are the probabilities of being in each state; to aid in the analysis of these quantities, I define some notation for them in the below Def.~\ref{def:PtjDef}. 
\begin{definition}
\label{def:PtjDef}
    Let \underline{$P_j^t \in [0,1]$} be the probability the Markov chain being considered is in the macrostate $j \in \mathbb{N}_0$ at time $t \in \mathbb{R}^+$.
\end{definition}
Def.~\ref{def:PtjDef} is used in Lemma~\ref{lem:asDiffEqn}, which derives a differential equation governing the dynamics of the above constructed model.
\begin{lemma}
\label{lem:asDiffEqn}
    The probabilities of being in each macrostate, $P_j^t$, of the Markov chain model presented in Sec.~\ref{sec:ModelStatement}, are governed by the differential equations:
    \begin{align}
        \label{eqn:diffInLemStatement}
         \dfrac{dP_j^t }{dt}
        &=
        \alpha_{j -1} P_{j-1}^{ t} + \alpha_{j +1} P_{j+1}^{ t} - 2\alpha_{j}P_j^{t},
    \end{align}
    and $\alpha_j$ is as in Def.~\ref{def:transitionDef}.
\end{lemma}
\begin{proof}
    Assume the discrete steps of the model defined in Sec.~\ref{sec:ModelStatement} have a constant duration of $\Delta t$. Let $\Delta P_j^t$ be the change in the probability of being in state $j \in \mathbb{N}_0$ during the $t$th step. Therefore, $\Delta P_j^t$ can be expressed as: $\forall j \in \mathbb{N}_0$,
    \begin{align}
        \Delta P_j^t 
        &=
        \bigg[ \alpha_{j -1} P_{j-1}^{ t-\Delta t} + \alpha_{j +1} P_{j+1}^{ t-\Delta t} - 2 \alpha_{j}P_j^{ t-\Delta t}\bigg] \Delta t
    \end{align}
    Therefore, dividing through by $\Delta t$ and taking the limit as $\Delta t \longrightarrow 0$:
    \begin{align}
    \lim_{\Delta t \longrightarrow 0} \bigg(
        \dfrac{\Delta P_j^t }{\Delta t} \bigg)
        &=
        \lim_{\Delta t \longrightarrow 0} \bigg(
        \alpha_{j -1} P_{j-1}^{ t-\Delta t} + \alpha_{j +1} P_{j+1}^{ t-\Delta t} - 2 \alpha_{j}P_j^{ t-\Delta t}
        \bigg)\\
        \Rightarrow
        \dfrac{dP_j^t }{dt}
        &=
        \alpha_{j -1} P_{j-1}^{ t} + \alpha_{j +1} P_{j+1}^{ t} - 2 \alpha_{j}P_j^{t}.
    \end{align}
    This is exactly the differential equation in the lemma statement.
\end{proof}
I note that for Eqn.~\ref{eqn:diffInLemStatement} to hold for \emph{all} values of $j \geq 0$ (i.e. all states in the Markov chain) requires:
\begin{enumerate}
    \item I additionally -- purely so the variables in Eqn.~\ref{eqn:diffInLemStatement} are always defined -- define, $\forall t \in \mathbb{R}^+$, $\alpha_{-1} = P_{-1}^t = 0$ and $\alpha_{\text{Max} +1} = P_{\text{Max} +1}^t = 0$. 
    \item Implicitly, $\alpha_{\text{Max}}$ is the only transition rate that depends on the direction of the transition (i.e. if the construction number is increasing or decreasing), if the transition is to state $\text{Max}+1$, $\alpha_{\text{Max}} = 0$. Otherwise, it is as in Def.~\ref{def:transitionDef}.
    \item The system of differential equations derived in Lemma~\ref{lem:asDiffEqn} (and given as Eqn.~\ref{eqn:diffInLemStatement}) is as large as the model is, and so for large models is large. It is useful to note that all equations in the system are linear and that there is the additional requirement that: $\forall t \in \mathbb{R}^+$,
\begin{align}
    \sum_{j \in \mathbb{N}_0} \bigg( P_j^{t} \bigg)
    =
    1
    \iff
    \sum_{j \in \mathbb{N}_0} \bigg( \dfrac{\partial P_j^{t}}{\partial t} \bigg)
    =
    0.
\end{align}
\end{enumerate}
\subsection{Asymptotic Behavior of the Model}
\label{sec:asmpto}
For clarity, in Sec.~\ref{sec:asmpto}, instead of denoting the probability of being in state $j$ at time $t$ as $P^t_j$, I express it as $\mathbb{P}\big( X_{t} = 0 \big)$. This, I believe, makes the use of Bayes' theorem~\cite{doi:10.1098/rstl.1763.0053} clearer in the below.

Consider any two times $t' \in \mathbb{R}^+$ and $t \in \mathbb{R}^+$, such that $t'>  t$. 

I can express the probability that the model is in state $j \in \mathbb{N}_0$ at time $t'$ as:\\
the probability the model is already in state $j$ at time $t$\\
\textbf{plus}\\
the probability it moves into state $j$, from any other state, after time $t$ but before time $t'$\\
\textbf{minus}\\
the probability it moves out of state $j$ after time $t$ but before time $t'$.\\

When $j = 0$, this analysis is slightly simpler as the probability of moving out of state $0$ at any time is zero and therefore $\mathbb{P}\big( X_{t'} = 0 \big)$ (i.e. $P^{t'}_0$) can be expressed as: $\forall t' > t \in \mathbb{R}^+$,
\begin{widetext}
\begin{align}
    \label{eqn:widetextProbs}
    \mathbb{P}\big( X_{t'} = 0 \big)
    &=
    \sum_{j \in \mathbb{W}'} \bigg( \mathbb{P}\big( X_{t'} = 0 \text{ } \vert \text{ } X_{t} = j \big) \mathbb{P}\big( X_{t} = j \big)  \bigg)
    =
    \mathbb{P}\big( X_{t'} = 0 \text{ } \vert \text{ } X_{t} = 0 \big) \mathbb{P}\big( X_{t} = 0 \big)
    +
    \sum_{j \in \mathbb{W}' \backslash \{ 0\}} \bigg( \mathbb{P}\big( X_{t'} = 0 \text{ } \vert \text{ } X_{t} = j \big) \mathbb{P}\big( X_{t} = j \big)  \bigg).
\end{align}
\end{widetext}
In a slight diversion, consider that, as $\forall j \in \mathbb{W}' \backslash \{ 0\}, \alpha_j > 0$,$ \forall t', t \in \mathbb{R}$ such that $t' - t >0$, 
\begin{align}
    \forall j \in \mathbb{W}' \backslash \{ 0\}&, \mathbb{P}\big( X_{t'} = 0 \text{ } \vert \text{ } X_{t} = j \big) > 0.
\end{align}
This implies that, as, $\forall j \in \mathbb{W}' \backslash \{ 0\}$, $\mathbb{P}\big( X_{t} = j \big) \geq 0$ and $\sum_{j = 0}^{\text{Max}} \big( \mathbb{P}\big( X_{t} = j \big) \big) = 1$: $\forall j \in \mathbb{W}' \backslash \{ 0\}$, 
\begin{align}
    \label{eqn:ProbsArePos}
    \sum_{j \in \mathbb{W}' \backslash \{ 0\}} \bigg( \mathbb{P}\big( X_{t'} = 0 \text{ } \vert \text{ } X_{t} = j \big) \mathbb{P}\big( X_{t} = j \big)  \bigg) > 0,
\end{align}
unless  $\forall j \in \mathbb{W}' \backslash \{ 0\}$, $\mathbb{P}\big( X_{t} = j \big) = 0$.\\
Therefore, returning to Eqn.~\ref{eqn:widetextProbs} and using Eqn.~\ref{eqn:ProbsArePos}, in addition to the fact that $\mathbb{P}\big( X_{t'} = 0 \text{ } \vert \text{ } X_{t} = 0 \big) = 1$,
\begin{align}
    \mathbb{P}\big( X_{t'} = 0 \big)
    &>
    \mathbb{P}\big( X_{t} = 0 \big),
\end{align}
unless  $\forall j \in \mathbb{W}' \backslash \{ 0\}$, $\mathbb{P}\big( X_{t} = j \big) = 0$. I.e. unless $\mathbb{P}\big( X_{t} = 0 \big) = 1$.

Therefore, I conclude that the probability of being in the state with construction number zero increases strictly monotonically until it reaches one. This, in turn, implies:
\begin{align}
    \lim_{t \longrightarrow \infty} \bigg( \mathbb{P}\big( X_{t} = 0 \big) \bigg) = 1.
\end{align}
Returning to the notation of earlier sections, this can equivalently be expressed as:
\begin{align}
    \lim_{t \longrightarrow \infty} \bigg( P^{t}_0 \bigg) = 1,
\end{align}
and therefore, for all other $j \not = 0$,
\begin{align}
    \lim_{t \longrightarrow \infty} \bigg( P^{t}_j \bigg) = 0.
\end{align}
I.e., as the time elapsed in the simulation tends to infinity, the probability that the simulation is in the macrostate with construction number zero tends to one and the probability that the simulation is in any macrostate with non-zero construction number tends to zero.

\section{Numerical Analysis of the Model}
While it would be of interest to consider the path the probabilities of being in each macrostate take to their asymptotic limits, in the interests of accessibility and appealing to a wider audience, I do not analytically solve this system of differential equations. Instead, I solve it numerically to demonstrate how the solution behaves. In some ways, this is superfluous as I have already shown the asymptotics of the model, which is all that matters (as per Sec.~\ref{sec:DynamicsTimelines}). However, it is useful to confirm these results numerically and also to examine the dynamics of the model in reaching this asymptotic model.
\subsection{System with 5 Macrostates, Starting in Macrostate 3}
The first plot is simple, just to get the lay of the land. It is presented in Fig.~\ref{fig:firstPlot_3_5} and numerically solves the system of differential equations for a system with five states (i.e. four time machines at most) and starting in the state with three time machines.
\begin{figure}[h!]
    \includegraphics[width=\textwidth/2]{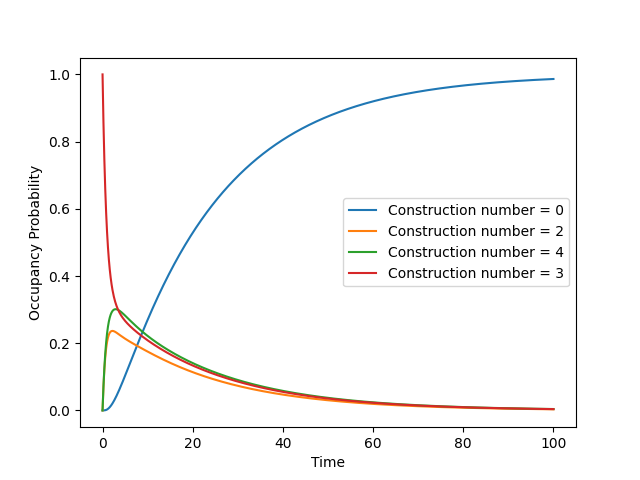}
    \caption{Occupancy probabilities as the Markov chain evolves of a five state Markov chain model, with the initial state being the one with a construction number of three.}
    \label{fig:firstPlot_3_5}
\end{figure}
From it, we can see that the probabilities of being in a state with a construction number other than zero decays to zero extremely quickly, while the probability of being in the state with a construction number of 0 quickly tends to 1. 

It is interesting to note that for macrostates other than the one with construction number zero or the initial macrostates, there is a brief increase in the probability of being in those states before it peaks and also decays extremely quickly.

\subsection{System with 500 Macrostates, Starting in Macrostate 3}
I then consider increasing the maximum allowed construction number. This is a more physical situation as -- in reality -- the maximum construction number will be very large. I do not change the initial state just yet, as I aim to isolate the single feature of the simulation I am changing. 

The results of this simulation are given in Fig.~\ref{fig:secondPlot_3_500}. The only change between the simulations in Fig.~\ref{fig:firstPlot_3_5} and Fig.~\ref{fig:secondPlot_3_500} is that the number of states has been increased from five to 500.
\begin{figure}[h!]
    \includegraphics[width=\textwidth/2]{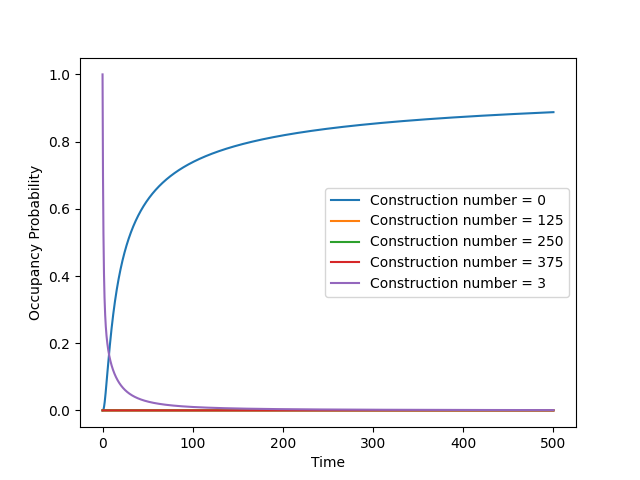}
    \caption{Occupancy probabilities as the Markov chain evolves of a 500 state Markov chain model, with the initial state being the one with a construction number of 3.}
    \label{fig:secondPlot_3_500}
\end{figure}

The first thing to note in Fig.~\ref{fig:secondPlot_3_500} is that the rate at which the probability of being in the state with construction number zero tends to one is much slower; the probability of being in the initial state (again, with a construction number of 3) still decays extremely quickly, which infers that the probability of being in non-initial states without a construction number of zero -- cumulatively -- stays larger for longer (although likely spread out among many of those states). As the initial macrostate is still relatively close to the zero macrostate, this infers that the probability of having a construction number greater than that of the initial macrostate actually increases initially before decaying.

\subsection{System with 500 Macrostates, Starting in Macrostate 100}
The final step to get a realistic Markov chain model is to increase the initial construction number (i.e. the construction number of the state that the Markov chain is initially in). This is more realistic as, although the initial macrostate is not fixed or derivable within my model, it is unlikely -- reasoning outside the model -- to be small, compared to the maximum possible construction number. Purely because relatively few of the macrostates have low construction numbers.

This simulation is shown in Fig.~\ref{fig:thirdPlot_100_500}.
\begin{figure}[h!]
    \includegraphics[width=\textwidth/2]{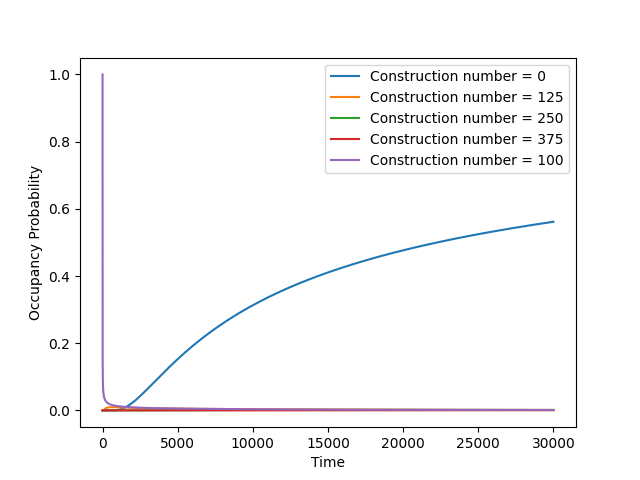}
    \caption{Occupancy probabilities as the Markov chain evolves of a 500 state Markov chain model, with the initial state being the one with a construction number of 100.}
    \label{fig:thirdPlot_100_500}
\end{figure}
It now, in Fig.~\ref{fig:thirdPlot_100_500}, takes much much longer for the probability of the Markov chain being in the state with a construction number of 0 to get close to one. This is the most striking consequence of this change.

This is likely due to the time it takes for the probability to percolate down -- through all the intermediate states -- to the state with construction number 0, but that \emph{is} the direction the probability -- overall -- flows in. The system still, at very long time-scales, behaves in the same way it does when the initial construction number is much lower. 

Another striking feature is how precipitous the decay of the probability of being in the initial state is. This is even faster than had been the case for when the initial construction number is lower.

A final thing to note in Fig.~\ref{fig:thirdPlot_100_500} is that states that are both not the initial state and do not have construction number one barely have any probably of being the state the Markov chain model is in at any time.

Further simulations are provided -- without commentary -- in Appendix~\ref{sec:MoreExamplesAppend}.
The code used to generate all the plots can be found on my Github~\cite{MyGithubCitation}.

\subsection{Interpretation of Results}
\label{sec:interpret}
The first key feature of the above simulations is that their asymptotic behavior is consistent, regardless of the number of states in the Markov chain model or the initial state of the model.
But the most striking and important feature of the above simulations is that, without exception, eventually the probability of being in the macrostate with construction number zero reaches one. All other macrostates eventually have probability zero.

Therefore, I conclude that, assuming my model, time travel is self-suppressing: the timeline is continually rewritten until it inevitably reaches a timeline with no time machines ever being constructed. At this point, no further changes to the timeline are possible. 

\section{Comparison to Other Suggested Explanations}
This paper is not the first to suggest an explanation for the lack of perceived time-travellers.
In fact, a variety of theories have been proposed to reconcile the theoretical permissibility of time travel with the absence of observed time-travellers. However, most assume physical or logical to time-traveller us meeting time travellers. 

The most obvious solution to us not seeing travellers from the future is also the least interesting: time travel may not be possible for technical reasons (i.e. physics does not permit travel backwards in time\footnote{Although results already exist to indicate reversing the flow of time is possible~\cite{jackson2023accreditation, jackson2025improvedaccreditationanaloguequantum}.}). It may simply not be possible to travel through time. This may be the true solution -- Occam's razor would suggest so -- but it is not definitively the case so other explanations are explored below and compared to the model originated herein.

Ref.~\cite{deutsch1991quantum} introduced a quantum model of closed timelike curves (i.e. paths through space-time where you can end up at an earlier point in time than where you started) where paradoxes are resolved through a consistency condition: time travel is allowed, but only if the events it causes are self-consistent. Similarly, Novikov's self-consistency conjecture~\cite{novikov1990self} also permits time travel under the constraint that any actions taken by time-travellers must preserve the timeline’s consistency (i.e. not cause paradoxes). While these models avoid paradoxes, which could imply time-travel is impossible or make it less likely, they do not explain the \emph{complete} absence of observable travellers unless \emph{any} notice of time-travellers causes us -- as a species -- to never develop time-travel.

Ref.~\cite{hawking1992chronology} instead argues that physical laws would prevent the formation of closed timelike curves altogether. This places the prohibition at the level of physical impossibility, unlike my model, which assumes no such prohibitions.

More recently, Ref.~\cite{greenberger2020time} considered both quantum mechanics and the many-worlds interpretation of it~\cite{sep-qm-manyworlds} to suggest that paradoxes are resolved across branching timelines. While this avoids logical contradictions, which may suppress time-travel, it does not predict the systematic complete erasure of time travel.

By contrast, the model developed in this paper introduces no physics-based barriers to time travel. Instead, it shows that time travel -- by enabling timeline alterations -- induces a dynamic instability that -- with very high probability -- leads to its own erasure. This self-suppressing mechanism results in the asymptotic convergence of all timelines toward states in which no time machines ever exist.
In some ways, in my model, the construction number of a macrostate can be seen as analogous to the energy of a thermodynamic macrostate of a physical object and how such systems have a tendency to end up in their lowest-energy state, like how my model ends in its lowest  construction number state.

\begin{Note}
    After the first version this paper had been uploaded to arXiv, I was made aware that Larry Niven had suggested an informal version of the conclusion of this paper in Ref.~\cite{allMyriad}, which he terms Niven's law:\\
    ``In the universe of discourse  permits the possibility of time travel and of changing the past, then no time machine will be invented in that universe."
    
\end{Note}

\section{Conclusions and Discussion}
As mentioned in Sec.~\ref{sec:interpret}, the key result of this paper is the asymptotic results of the model constructed herein: eventually, the model ends up -- and stays -- in a state with zero time machines (i.e. in the macrostate with construction number zero).

As was argued in Sec.~\ref{sec:DynamicsTimelines}, the asymptotic results are all that matters to those of us not travelling, as the continual re-writing of the timeline is imperceptible to us and we ``jump" straight to the asymptotic limit. This and the above mentioned asymptotic behaviour of the constructed model, mean that we should not expect to receive visitors from the future.   

Therefore, even if time travel is not technically forbidden by physics or mathematics or metaphysics, we should not expect to see time-travellers. Nor should we expect time travel to be an anticipated aspect of future societies, unless the assumptions made herein can be evaded in some way.

It is note-worthy that the model constructed herein has a free parameter, $\beta$, which scales the transition rates; this was shown analytically to not affect the asymptotic, i.e. most important, behavior of this model (assuming $\beta \not = 0$). It is therefore immaterial. However, it is possible that this parameter could, by caution or some other factor, be driven extremely low, but this does not affect the asymptotic behaviour unless $\beta$ is driven to exactly zero.

I do not necessarily preclude $\beta$ being zero: its not impossible to imagine future societies or nefarious factions striving to change the past and maintain their changes by then enforcing the condition that $\beta = 0$.
To this end, it is important to note that the value of $\beta$ is not solely determined by physics; it is more of a sociological measure, so it is theoretically possible for some regime to keep the value of $\beta$ low -- if not zero. 

I now turn to discuss the assumptions of this model and some of their potential failings. Most prominent, to me, of these assumptions -- that may fail -- is the assumption that the orthogonal time can or will tend to infinity. No car can run forever, so why should we expect a time machine to?

To consider the future of machine lifespans, it is useful to first consider the past. How has the useful lifetime of large complex machines changed?
Perhaps the closest analogy is cars: cars have -- in the last 50 years -- had their lifetimes extended significantly~\cite{Ford_2012, Harley_2023}, as technology has improved to enable them to keep running longer.

If these trends continue, in 1000 years or so, complex machines will have lifespans so long that they can be well approximated by saying they are infinite. This does not factor in potential breakthroughs in technological maintenance (e.g. mimicking how biological systems can repair and maintain themselves, allowing them to run indefinitely), just incremental improvements based on historical patterns.

There is also the possibility that time machines can return to their original time for maintenance or to be cannibalized to build new time machines (the main argument for construction numbers being limited by some upper bound was limitations on resources and time, but a working time machine mitigates both these issues).
Assuming this is possible, provided that where/when it was constructed is still in the timeline, a time machine can always return to get repairs and new crew and whatever other resources are needed. Depending on your view on paradoxes, the same resources can possibly be used repeatedly by re-writing the timeline to keep obtaining the resources.

The model providing these results is, admittedly, speculative but its only purpose is to provide an alternative to a lack of recorded time-travellers precluding the possibility of time-travel. Ultimately, this is panacea to most issues that may be found with this paper: its just a potential solution, intended to provide an alternative to concluding that the lack of observed time-travellers totally precludes the possibility of time travel.
It can also be viewed as just an approximation of what may happen; perhaps future work will relax the assumptions to deal with any issues you find?

\section{Acknowledgements}
I would like to thank James Lyons-Weiler, for extensive comments on the first version of this paper, and James Smirle, for making me aware of Niven's Law. 

\newpage
\bibliography{References}

\onecolumngrid
\newpage
\appendix

\section{More Examples of Solutions of the System of Differential Equations}
\label{sec:MoreExamplesAppend}
\subsection{Markov Chain Model with 1000 States and an Initial Construction Number of 50}
\begin{figure}[H]
    \begin{center}
    \includegraphics[width=0.6\textwidth]{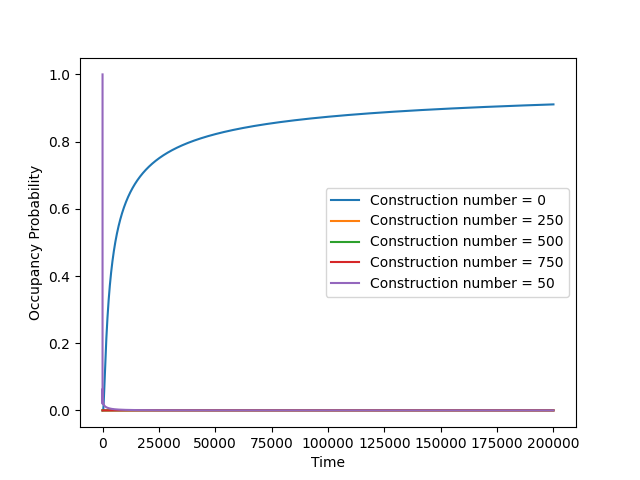}
     \end{center}
    \caption{Occupancy probabilities as the Markov chain evolves of a 1000 state Markov chain model, with the initial state being the one with a construction number of 50.}
    \label{fig:appenFirstPlot_50_1000}
\end{figure}
\subsection{Markov Chain Model with 400 States and an Initial Construction Number of 399}
\begin{figure}[H]
    \begin{center}
    \includegraphics[width=0.65\textwidth]{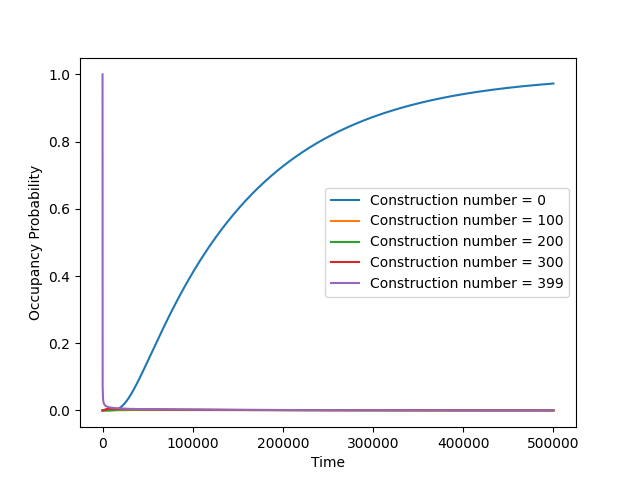}
    \end{center}
    \caption{Occupancy probabilities as the Markov chain evolves of a 400 state Markov chain model, with the initial state being the one with a construction number of 399.}
    \label{fig:appenSecondPlot_399_400}
\end{figure}
\subsection{Markov Chain Model with 200 States and an Initial Construction Number of 20}
\begin{figure}[H]
    \begin{center}
    \includegraphics[width=0.65\textwidth]{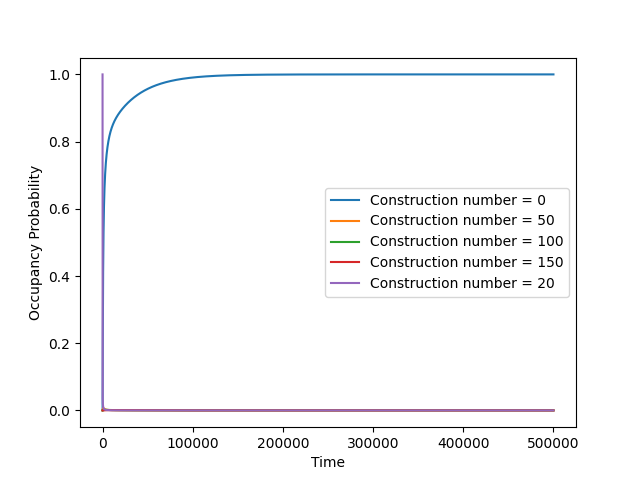}
    \end{center}
    \caption{Occupancy probabilities as the Markov chain evolves of a 200 state Markov chain model, with the initial state being the one with a construction number of 20.}
    \label{fig:appenSecondPlot_20_200}
\end{figure}
\subsection{Markov Chain Model with 100 States and an Initial Construction Number of 0}
\begin{figure}[H]
    \begin{center}
    \includegraphics[width=0.65\textwidth]{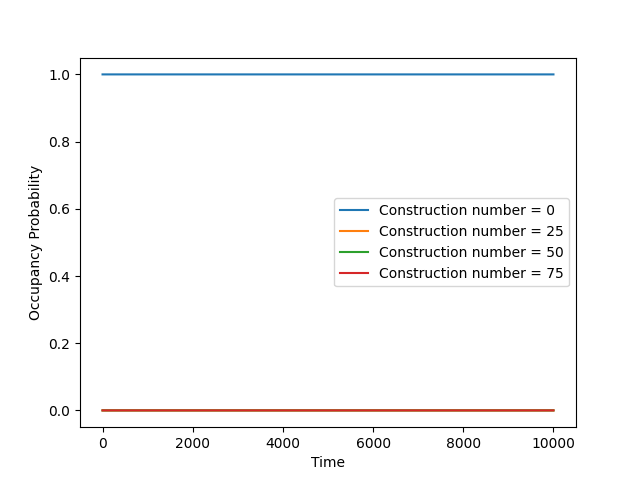}
    \end{center}
    \caption{Occupancy probabilities as the Markov chain evolves of a 100 state Markov chain model, with the initial state being the one with a construction number of 0.}
    \label{fig:appenSecondPlot_0_100}
\end{figure}
\subsection{Markov Chain Model with 10 States and an Initial Construction Number of 5}
\begin{figure}[H]
    \begin{center}
    \includegraphics[width=0.65\textwidth]{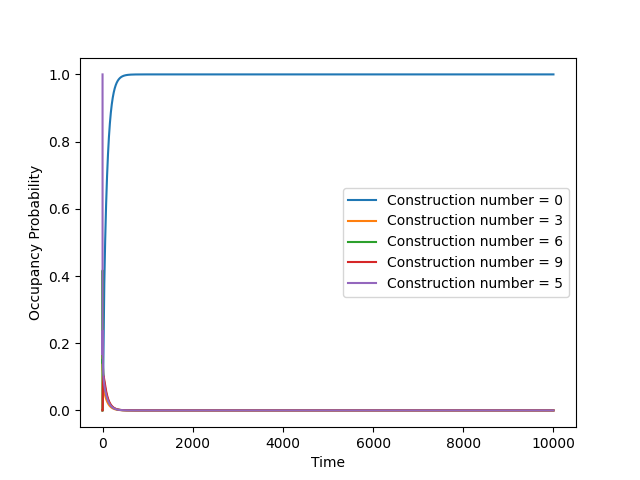}
    \end{center}
    \caption{Occupancy probabilities as the Markov chain evolves of a 10 state Markov chain model, with the initial state being the one with a construction number of 5.}
    \label{fig:appenSecondPlot_5_10}
\end{figure}
\begin{figure}[H]
    \begin{center}
    \includegraphics[width=0.65\textwidth]{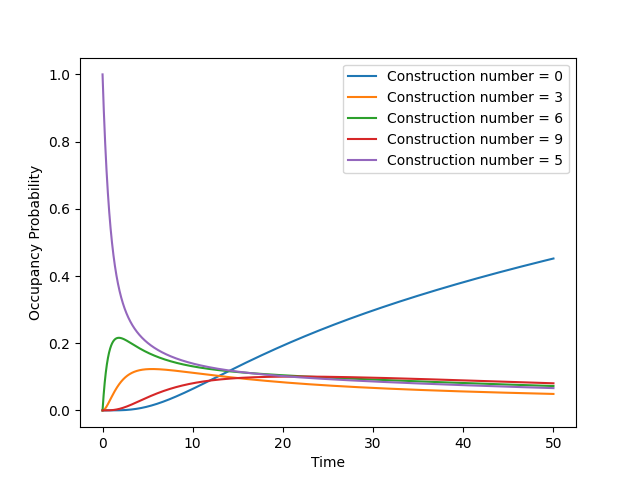}
    \end{center}
    \caption{Occupancy probabilities as the Markov chain evolves of a 10 state Markov chain model, with the initial state being the one with a construction number of 5. Plotted over a shorter timeframe}
    \label{fig:appenSecondPlot_5_10_SHORTER}
\end{figure}
\end{document}